\documentclass[preprint,12pt]{elsarticle}

\usepackage{amssymb}
\usepackage{amsmath}
\usepackage{amsthm}
\newtheorem{theorem}{Theorem}
\newtheorem{lemma}[theorem]{Lemma}
\usepackage{booktabs}
\usepackage{algorithm}
\usepackage{algpseudocode}

\journal{Arxiv}

\begin{document}

\begin{frontmatter}



\title{New Theoretical Insights and Algorithmic Solutions \\for Reconstructing Score Sequences from Tournament Score Sets}


\author{Bowen Liu} 
\ead{liubowenmbu@outlook.com}
\ead[url]{https://orcid.org/0009-0008-7697-0229} 

\affiliation{organization={Department of Mathematics, Shenzhen MSU-BIT University},
            addressline={Shenzhen MSU-BIT University, Shenzhen, China.}, 
            city={Shenzhen},
            postcode={518172}, 
            state={Guangzhou},
            country={China}}

\begin{abstract}
The \text{score set} of a tournament is defined as the set of its distinct out-degrees. In 1978, Reid proposed the conjecture that for any set of nonnegative integers $D$, there exists a tournament $T$ with a degree set $D$. In 1989, Yao presented an arithmetical proof of the conjecture, but a general polynomial-time construction algorithm is not known. This paper proposes a necessary and sufficient condition and a separate necessary condition, based on the existing Landau's theorem for the problem of reconstructing score sequences from score sets of tournament graphs. The necessary condition introduces a structured set that enables the use of group-theoretic techniques, offering not only a framework for solving the reconstruction problem but also a new perspective for approaching similar problems. In particular, the same theoretical approach can be extended to reconstruct valid score sets given constraints on the frequency of distinct scores in tournaments. Based on these conditions, we have developed three algorithms that demonstrate the practical utility of our framework: a polynomial-time algorithm and a scalable algorithm for reconstructing score sequences, and a polynomial-time network-building method that finds all possible score sequences for a given score set. Moreover, the polynomial-time algorithm for reconstructing the score sequence of a tournament for a given score set can be used to verify Reid's conjecture. These algorithms have practical applications in sports analysis, ranking prediction, and machine learning tasks such as learning-to-rank models and data imputation, where the reconstruction of partial rankings or sequences is essential for recommendation systems and anomaly detection.
\end{abstract}


\begin{highlights}
    \item \textbf{Theoretic framework}: Introduces a structured approach to score sequence reconstruction, providing both necessary and sufficient conditions.  
    \item \textbf{Efficient algorithms}: Develops two polynomial-time algorithms and a scalable method for reconstructing score sequences in tournaments.  
    \item \textbf{Heuristic optimization}: The scalable algorithm can be used as a heuristic for large-scale ranking and optimization problems.  
\end{highlights}


\end{frontmatter}



\section{Introduction}

This paper explores the problem of correctly reconstructing score sequences from score sets of tournament graphs. In 1978, Reid \cite{reid1978scoresets} proposed the conjecture that for any set of nonnegative integers $D$, there exists a tournament $T$ whose degree set is $D$. In 1989, Yao \cite{yao1989reidconjecture} presented an arithmetical proof of the conjecture. Previous research shows that while specific algorithms can construct score sequences for a given small-scale score set, a complete polynomial-time algorithm for larger-scale cases remains unknown \cite{Reid86}.

This paper proposes a necessary and sufficient condition and a necessary condition based on Landau's existing theorem for correctly reconstructing the score sequence from a score set of a tournament graph. These conditions introduce a structured set that enables the use of group-theoretic techniques, and the solution space of the reconstruction problem forms a subset of this structured set. Moreover, the necessary and sufficient condition can be seen as a generalization of the upper bound given in Lemma 6 of \cite{Reid86}, with detailed explanation provided in Section \ref{sec:pol}. Beyond providing a rigorous framework for solving the reconstruction problem, these conditions also offer new perspectives for approaching similar problems. In particular, the same theoretical approach can be extended to reconstruct valid score sets given constraints on the frequency of distinct scores in tournaments. Based on these conditions, we have developed three algorithms that demonstrate the practical utility of our framework: a polynomial-time algorithm using dynamic programming, a scalable algorithm, and a polynomial-time network-building method that finds all possible score sequences for a given score set. Moreover, the scalable algorithm can be regarded as a special case of the Hole-Shift algorithm from \cite{Reid86} when the difference between the maximum and minimum elements in the score set is 1. These developments aim to fill existing research gaps and advance the theory of tournament graphs.

Beyond the theoretical aspect, this problem has practical applications in sports tournaments, particularly in leagues where each team plays against every other team, such as in football or chess \cite{Suksompong2021}. Reconstructing score sequences can help predict final rankings, assess the fairness of schedules, and determine the relative strengths of teams, which is useful for coaches and analysts \cite{Damkhi2018}. In the field of machine learning, the reconstruction of score sequences is related to ranking learning and data completion, where the goal is to infer global rankings from partial information. For example, from possible score sequences, we can use ranking methods such as the method described in \cite{moon1968topics} to predict rankings. Even when match results are incomplete, applying these methods allows us to estimate the relative standings of participants. This has direct applications in recommendation systems, where we predict a user’s preference sequence based on their total ratings, improving the system’s accuracy. Additionally, anomaly detection can benefit from analyzing possible score sequences to identify inconsistencies or noise in data. For instance, Miao \cite{MIAO2024103569} demonstrates how reconstruction-based methods can effectively identify anomalies by using contrastive generative adversarial networks to enhance detection accuracy.

By exploring both the theoretical and practical sides of this problem, this research contributes to the understanding of tournament graphs and their broader real-world implications, including machine learning and AI-driven tasks.

The paper is organized as follows. The necessary preparatory materials for starting the paper are in the Section \ref{sec:Prel}. The necessary and sufficient condition, as well as the necessary condition, are discussed in Section \ref{sec:condition}. The benefits of these conditions are outlined in Sections \ref{sec:benefit1} and \ref{sec:benefit2}. Polynomial-time algorithms are presented in Section \ref{sec:pol}, including one that reconstructs score sequences and another that builds a network to find all possible score sequences for a given score set, along with their experimental results. A scalable algorithm for reconstructing score sequences, along with its experimental results, is discussed in Section \ref{sec:fast}. Finally, the conclusions follow in Section \ref{sec:con}. The C++ implementation of the proposed algorithms is available from the authors upon request.

\section{Preliminaries}
\label{sec:Prel}

To introduce this paper, we begin by reviewing the concept of tournament graphs and illustrating their structure with examples. A tournament graph is a directed graph formed by assigning a direction to each edge in a complete graph. This means that for any two distinct vertices, there is exactly one directed edge connecting them. Tournament graphs naturally model competitive scenarios, where vertices represent players or teams, and a directed edge from vertex $u$ to vertex $v$ signifies that $u$ has defeated $v$.

Below, we present an example of a tournament graph \ref{fig:tour_small}. The score sequence of a tournament graph is defined as the sequence of out-degrees of its vertices, indicating the number of wins each player has achieved.

\begin{figure}[tbhp]
\centering
\includegraphics[width=0.2\textwidth]{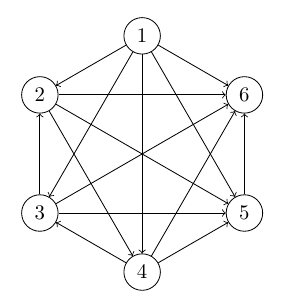}
\caption{An example of a tournament graph.}
\label{fig:tour_small}
\end{figure}

Building on this foundation, we next review Landau’s theorem and its reformulation, as they serve as the basis for the results developed in this paper.

\begin{theorem}[Landau \cite{Landau53_III}]\label{thm:lan}
 A non-decreasing sequence of non-negative integers 
$S = s_1, s_2, \ldots, s_m$ is a score sequence of an $m$-tournament if and only if
\begin{displaymath}
\sum_{i=1}^k s_i \geq \frac{k(k-1)}{2}, \quad \forall k \in \mathbb{N}, 1 \leq k \leq m,
\end{displaymath}
with equality when $k = m$.
\end{theorem}

\begin{proof}
 See \cite{Landau53_III}.
\end{proof}

To illustrate this, consider Fig. \ref{fig:tour_small}, where the non-decreasing sequence of non-negative integers $S$ is given by $\{0,1,3,3,3,5\}$. This sequence satisfies Landau’s theorem. For instance, when $k=3$, the left-hand side of the inequality, represented as $\sum_{i=1}^{k} s_i $, equals $4$, while the right-hand side, given by $\frac{k(k-1)}{2}$, equals $3$, ensuring that $4 \geq 3$.

Now, we consider only the case where $n \geq 2$.  
Let the degree set be $D = \{\alpha_i \mid \alpha_i \in \mathbb{N}, \alpha_i \geq 0, i = \in \{1,\dots, n\} \}$, where the elements \(\alpha_i\) are arranged in increasing order.  
Define the exponent set of the score as \(E = \{x_i \mid x_i \in \mathbb{N}, x_i \geq 1, i = \in \{1,\dots ,n\} \}\), where each \(x_i\) represents the number of times the score \(\alpha_i\) appears in a given score sequence. 

For example, in Fig. \ref{fig:tour_small}, the degree set is $D=\{0,1,3,5\}$, and the corresponding exponent set is $E=\{1,1,3,1\}$, indicating the number of times each score appears.

We define:  
\[
p_i = \sum_{k=1}^{i} x_k, \quad \text{where } p_0 = 0; \quad
q_i = \sum_{k=1}^{i} \alpha_k \cdot x_k, \quad \text{where } q_0 = 0.
\]  

Note that in this paper, the terms "degree set" and "score set" refer to the same concept.

Now, let's examine another reformulation of Landau's theorem, which also appears in \cite{Ivanyi2013}.

\begin{theorem}\label{thm:2.2}
Theorem \ref{thm:lan} can be reformulated as follows:
  A score sequence, characterized by a degree set $D$ and an exponent set $E$,  represents the score sequence of a tournament if and only if
\begin{displaymath}
\sum_{i=1}^k \alpha_i \cdot x_i \geq \frac{\sum_{i=1}^{k} {x_i}(\sum_{i=1}^{k} {x_i}-1)}{2}, \quad \forall k \in \mathbb{N}, 1 \leq k \leq n,
\end{displaymath}
with equality holding when $k = n$.
\end{theorem}

\begin{proof}

Assume that the score sequence defined by \(D\) and \(E\) satisfies Theorem \ref{thm:2.2} but does not satisfy Theorem \ref{thm:lan}. Then, \(\exists k \in \mathbb{N}\) such that 
\begin{displaymath}
\sum_{i=1}^{k-1} \alpha_i \cdot x_i + \sum_{i=1}^{x} \alpha_k < \frac{(\sum_{i=1}^{k-1} {x_i}+x)(\sum_{i=1}^{k-1} {x_i}+x-1)}{2}, \quad \forall x \in \mathbb{N}, 0<x<x_k, 
\end{displaymath}
yet for \(x=0\) and \(x=x_k\), the conditions of Theorem \ref{thm:2.2} are satisfied. Considering the graphs of the functions \(\alpha_k \cdot x + b\) and \(\frac{(x+c)(x+c-1)}{2}\) (here \(b=\sum_{i=1}^{k-1} \alpha_i \cdot x_i, \quad c=\sum_{i=1}^{k-1} {x_i}\) are constants), it is evident that for any \(0 < x < x_k\), the inequality \(\alpha_k \cdot x + b > \frac{(x+c)(x+c-1)}{2}\) holds. Therefore, the assumption is incorrect, and the score sequence defined by \(D\) and \(E\) does satisfy \ref{thm:2.2}, then it must also satisfy Theorem \ref{thm:lan}.

Furthermore, since Theorem \ref{thm:2.2} consists only of a subset of the inequalities in Theorem \ref{thm:lan}, along with the final equality, if the score sequence defined by \(D\) and \(E\) satisfies Theorem \ref{thm:lan}, then it also satisfies Theorem \ref{thm:2.2}.

\end{proof}


\section{Conditions for the Existence of the Score Sequence}
\label{sec:condition}

In this section, we present both the necessary and sufficient condition, as well as the necessary condition, based on the existing Landau's theorem for the problem of reconstructing score sequences from score sets of tournament graphs. The following two theorems correspond to the necessary and sufficient condition and the necessary condition proposed in this paper. Additionally, we introduce an extended theorem for the necessary condition from a group-theoretic perspective.

These conditions not only provide theoretical insights into the structure of score sequences but also facilitate efficient algorithmic design. The implications of these conditions on computational complexity and practical reconstruction methods will be discussed in Section \ref{sec:benefit1} and Section \ref{sec:benefit2}.

\begin{theorem}\label{thm:2.3}
Assume there exists a degree set $D$ and an exponent set $E$ that satisfy the following condition:
\begin{displaymath}
\sum_{i=1}^k \alpha_i \cdot x_i \geq \frac{\sum_{i=1}^{k} {x_i}(\sum_{i=1}^{k} {x_i}-1)}{2}, \quad \forall k \in \mathbb{N}, 1 \leq k \leq n,
\end{displaymath}
with equality when $k = n$, if and only if they satisfy this condition: 

$\forall k \in \mathbb{N}, 1 \leq k \leq n,$
\begin{equation}\label{eq:2.1}
1 \leq x_k \leq \frac{(2(\alpha_k-p_{k-1})+1)+\sqrt{(2(\alpha_k-p_{k-1})+1)^2+8(q_{k-1}-p_{k-1}(p_{k-1}-1)/2)}}{2},  
\end{equation}
\begin{displaymath}
and\quad q_{k-1}-p_{k-1}(p_{k-1}-1)/2 \geq 0, 
\end{displaymath}
with the right side of the inequality becoming equal when $k = n$.
\end{theorem}

\begin{proof}

\begin{displaymath}
\forall k \in \mathbb{N}, 1 \leq k \leq n, 
\sum_{i=1}^k \alpha_i \cdot x_i \geq \frac{\sum_{i=1}^{k} {x_i}(\sum_{i=1}^{k} {x_i}-1)}{2}
\end{displaymath}
\begin{displaymath}
\Leftrightarrow 2(\sum_{i=1}^k \alpha_i \cdot x_i) \geq (\sum_{i=1}^{k} {x_i})^2-\sum_{i=1}^{k} {x_i}
\end{displaymath}
\begin{displaymath}
\Leftrightarrow 2\cdot q_{k-1} +2\alpha_k \cdot x_k \geq (p_{k-1}+x_k)^2-p_{k-1}-x_k
\end{displaymath}
\begin{displaymath}
\Leftrightarrow 2\cdot q_{k-1} +2\alpha_k \cdot x_k \geq {p_{k-1}}^2+{x_k}^2+2{x_k{p_{k-1}}}-p_{k-1}-x_k
\end{displaymath}
\begin{displaymath}
\Leftrightarrow 0 \geq {x_k}^2-(2(\alpha_k-p_{k-1})+1)x_k-2(q_{k-1}-{p_{k-1}}(p_{k-1}-1)/2).
\end{displaymath}

Define $\triangle$ as follows:
\begin{equation}\label{eq:del}
\triangle = (2(\alpha_k-p_{k-1})+1)^2+8(q_{k-1}-{p_{k-1}}(p_{k-1}-1)/2).
\end{equation}

From the first part of Theorem \ref{thm:2.3}, $\sum_{i=1}^k \alpha_i \cdot x_i \geq \frac{\sum_{i=1}^{k} {x_i}(\sum_{i=1}^{k} {x_i}-1)}{2}$, $\forall k \in \mathbb{N}, 1 \leq k \leq n$, it follows that for $k \geq 2$, $q_{k-1}-{p_{k-1}}(p_{k-1}-1)/2 \geq 0$, thus $\triangle \geq 0$. For $k = 1$, it follows from the hypothesis that both $p_0$ and $q_0$ are zero. Moreover, $q_{k-1}-{p_{k-1}}(p_{k-1}-1)/2 = 0$, which implies $\triangle \geq 0$.

From the second part of Theorem \ref{thm:2.3}, we also have $\triangle \geq 0$. Therefore: 

\begin{displaymath}
\frac{(2(\alpha_k-p_{k-1})+1)-\sqrt{\triangle}}{2} \leq x_k \leq \frac{(2(\alpha_k-p_{k-1})+1)+\sqrt{\triangle}}{2}.
\end{displaymath}

As previously established, we always have $q_{k-1}-{p_{k-1}}(p_{k-1}-1)/2 \geq 0$. Therefore, $\frac{(2(\alpha_k-p_{k-1})+1)-\sqrt{\triangle}}{2}<0$.

From the left-hand side of Theorem \ref{thm:2.3}, it follows that $\exists x_k \geq 1, x_k\in \mathbb{N}$ satisfying $\frac{(2(\alpha_k-p_{k-1})+1)-\sqrt{\triangle}}{2} \leq x_k \leq \frac{(2(\alpha_k-p_{k-1})+1)+\sqrt{\triangle}}{2}$. Thus, we obtain the right-hand side of Theorem \ref{thm:2.3}. From the right-hand side of Theorem \ref{thm:2.3}, it follows that there exists $x_k$ satisfying $1 \leq x_k \leq \frac{(2(\alpha_k-p_{k-1})+1)+\sqrt{\triangle}}{2}$, thus $x_k\in \mathbb{N}$ satisfies $\frac{(2(\alpha_k-p_{k-1})+1)-\sqrt{\triangle}}{2} \leq x_k \leq \frac{(2(\alpha_k-p_{k-1})+1)+\sqrt{\triangle}}{2}$ and $x_k \geq 1$. Thus, we obtain the left-hand side of Theorem \ref{thm:2.3}.

\end{proof}

\begin{theorem}\label{thm:2.4}
Assume there exists a degree set $D$ and an exponent set $E$ that satisfy the following condition:
$\forall k \in \mathbb{N}, 1 \leq k \leq n,$
\begin{displaymath}
1 \leq x_k \leq \frac{(2(\alpha_k-p_{k-1})+1)+\sqrt{(2(\alpha_k-p_{k-1})+1)^2+8(q_{k-1}-p_{k-1}(p_{k-1}-1)/2)}}{2},  
\end{displaymath}
\begin{displaymath}
and\quad q_{k-1}-p_{k-1}(p_{k-1}-1)/2 \geq 0. 
\end{displaymath}
with equality holding on the right-hand side when $k = n$, then they satisfy this condition: 
\begin{displaymath}
\exists f\in \mathbb{N},f\leq (2\alpha_n+1)/2,
\end{displaymath}
\begin{displaymath}
(2\alpha_n+1-2f)f=\sum_{i=1}^{n-1} {(\alpha_n-\alpha_i)x_i}. 
\end{displaymath}
\end{theorem}

\begin{proof}

\begin{displaymath}
x_n = \frac{(2(\alpha_n-p_{n-1})+1)+\sqrt{(2(\alpha_n-p_{n-1})+1)^2+8(q_{n-1}-p_{n-1}(p_{n-1}-1)/2)}}{2},
\end{displaymath}
\begin{displaymath}
then\quad \sqrt{(2(\alpha_n-p_{n-1})+1)^2+8(q_{n-1}-p_{n-1}(p_{n-1}-1)/2)} \quad is \quad odd \quad because
\end{displaymath}
\begin{displaymath}
2(\alpha_n-p_{n-1})+1 \quad is \quad odd,\quad and \quad only \quad even \quad numbers \quad are \quad divisible \quad by \quad 2.
\end{displaymath}
Let $\beta\geq 0$, 
\begin{subequations}
\begin{align}
    \beta^2
    &=(2(\alpha_n-p_{n-1})+1)^2+8(q_{n-1}-p_{n-1}(p_{n-1}-1)/2)\\
    &=(2\alpha_n+1-2p_{n-1})^2+8(q_{n-1}-p_{n-1}(p_{n-1}-1)/2)\\
    &=(2\alpha_n+1)^2 +4p_{n-1}^2-4(2\alpha_n+1)p_{n-1}-4p_{n-1}(p_{n-1}-1)+8q_{n-1}\\
    &=(2\alpha_n+1)^2-8\alpha_n p_{n-1}+8q_{n-1}\\
    &=(2\alpha_n+1)^2-8(\sum_{i=1}^{n-1} {\alpha_n x_i}-8\sum_{i=1}^{n-1} {\alpha_i x_i})\\
    &=(2\alpha_n+1)^2-8 \sum_{i=1}^{n-1} {(\alpha_n-\alpha_i)x_i}\label{eq:2.2f}
\end{align}
\end{subequations}
\begin{subequations}
\begin{align}
    \sum_{i=1}^{n-1} {(\alpha_n-\alpha_i)x_i}
    &=\frac{(2\alpha_n+1)^2-\beta^2}{8}
    &=\frac{(2\alpha_n+1-\beta)(2\alpha_n+1+\beta)}{8}
\end{align}
\end{subequations}
Since $\beta^2$ is odd, it follows that $\beta$ must also be odd. $\sum_{i=1}^{n-1} {(\alpha_n-\alpha_i)x_i}$ is an integer, $\beta\leq (2\alpha_n+1)$ from inequality \ref{eq:2.2f} and $\forall x_i, i\in \{1,2,\dots,n-1\}, x_i\geq 1$. So $\exists l \in \mathbb{N}$,
\begin{displaymath}
    (2\alpha_n+1-\beta)=4l \quad or \quad(2\alpha_n+1+\beta)=4l
\end{displaymath}
\begin{displaymath}
    and \quad 0\leq l\leq (2\alpha_n+1)/2
\end{displaymath}
Substituting any equation that satisfies these conditions, we obtain 
\begin{equation}
(2\alpha_n+1-2l)l=\sum_{i=1}^{n-1} {(\alpha_n-\alpha_i)x_i}. 
\end{equation}
Let $f=l$, then 
\begin{displaymath}
(2\alpha_n+1-2f)f=\sum_{i=1}^{n-1} {(\alpha_n-\alpha_i)x_i}. 
\end{displaymath}

\end{proof}

From Theorem \ref{thm:2.4}, it is easy to observe that this problem can be formulated as a positive linear programming problem. Directly solving it is inherently complex. However, if we consider each element as an element in the positive residue class group modulo $\alpha_n-\alpha_{n-1}$, then the right-hand side of the equation forms a finite set with at most $\alpha_n-\alpha_{n-1}$ elements.

Once certain $x_i$ values are determined, the remaining elements, if unrestricted, can generate a subgroup composed of these remaining elements. Moreover, the minimal generating set of this subgroup can be conveniently identified. Based on this idea, we provide a theoretical proof to support this perspective.

Assume $sol=\{s_i\mid s_i=((2\alpha_n+1-2i)i)\mod{(\alpha_n-\alpha_{n-1})}, i=\in\{0,1, \dots,\alpha_n-\alpha_{n-1}-1\}\}$ and $grp=\{g_i\mid g_i=(\alpha_n-\alpha_i)\mod{(\alpha_n-\alpha_{n-1})}, i=\in\{1,2,\dots,n-1\}\}$.

\begin{theorem}\label{thm:2.5}
If $\exists f\in \mathbb{N},f\leq (2\alpha_n+1)/2$, $(2\alpha_n+1-2f)f=\sum_{i=1}^{n-1} {(\alpha_n-\alpha_i)x_i}$, then $\sum_{i=1}^{n-1} {g_i x_i}\mod{(\alpha_n-\alpha_{n-1})} \in sol$.
\end{theorem}

\begin{proof}
\begin{equation}
    ((2\alpha_n+1-2f)f)\mod(\alpha_n-\alpha_{n-1})=\gamma
\end{equation}
\begin{subequations}
\begin{align}
\gamma&=(\sum_{i=1}^{n-1} {(\alpha_n-\alpha_i)x_i})\mod(\alpha_n-\alpha_{n-1})\\
&=(\sum_{i=1}^{n-1} {((\alpha_n-\alpha_i)\mod(\alpha_n-\alpha_{n-1}))x_i})\mod(\alpha_n-\alpha_{n-1})\\
&=(\sum_{i=1}^{n-1} {g_i x_i})\mod(\alpha_n-\alpha_{n-1})
\end{align}
\end{subequations}

Because of $((2\alpha_n+1-2f)f)\mod(\alpha_n-\alpha_{n-1})\in sol$, $\sum_{i=1}^{n-1} {g_i x_i}\mod{(\alpha_n-\alpha_{n-1})} \in sol$.
\end{proof}

\section{Benefits of Necessary and Sufficient Condition}
\label{sec:benefit1}

The necessary and sufficient condition classifies the combinations of \( x_i \) for all \( i < k \) into equivalence classes that depend only on two variables, \( p_{k-1} \) and \( q_{k-1} \). By leveraging the integrality of these two variables and the fact that the combinatorial structure is independent of \( n \), we can design a polynomial-time algorithm based on dynamic programming. Additionally, considering that the number of valid sequences of a tournament with $n$ vertices, as given in \cite{Kleitman1970}, is \( O(4^n / n^{5/2}) \) and experimental results show that the values of \( p_{k-1} \) and \( q_{k-1} \) reachable from previous \( x_i \) combinations are highly sparse—far fewer than the total number of possible combinations—the algorithm based on this necessary and sufficient condition can terminate efficiently when using an appropriate storage structure. Moreover, this condition provides a criterion with $O(n)$ complexity for quantitatively evaluating whether the current choice of $x_k$ has the potential to be part of a valid score sequence.

By setting \( x_i = 1 \) for all \( i \geq k+1 \) and utilizing the known values \( p_{k-1} \) and \( q_{k-1} \), we derive \( n-k \) inequalities that depend solely on \( x_k \):  
\[
1 \leq \frac{(2(\alpha_i - P_i) + 1) + \sqrt{(2(\alpha_i - P_i) + 1)^2 + 8(Q_i - P_i(P_i - 1)/2)}}{2},
\]
where  
\[
P_i = p_{k-1} + x_k + i - (k+1), \quad Q_i = q_{k-1} + x_k \cdot \alpha_k + \sum_{j=k+1}^{i-1} \alpha_j.
\]

This approach allows us to refine the bounds for \( x_k \) by solving a set of one-variable inequalities.  

It is important to emphasize that not every \( x_k \) within the range specified by Theorem \ref{thm:2.3} can necessarily be part of a valid score sequence. If \( x_k \) fails to satisfy the newly derived constraints, then despite meeting the conditions in Theorem \ref{thm:2.3}, there will exist at least one subsequent \( x_i \) (for \( i \geq k+1 \)) such that no assignment of the remaining \( x_j \) (for \( j \geq k+1, j \neq i \)) can fulfill Theorem \ref{thm:2.3}. Consequently, constructing a valid score sequence would become impossible.  

By leveraging this refinement, we can further narrow the feasible choices for \( x_k \) and identify those that are more likely to contribute to a valid score sequence.

To illustrate this phenomenon, consider the tournament with degree set \( D = \{2, 4, 7, 14\} \). From this set \( D \), we aim to reconstruct an exponent set \( E \) that satisfies Theorem \ref{thm:2.3}. When evaluating \( x_1 \), the formula derived earlier gives the range \( 1 \leq x_1 \leq 5 \). However, if \( x_1 = 5 \), the formula implies \( x_2 \leq 0 \), which is clearly invalid since \( x_2 \) must be greater than $1$.  

Furthermore, as shown in the table below, there is no valid score sequence where \( x_1 \geq 4 \). (Table \ref{tab:degree_set} all possible exponent sets \( E \) for the degree set \( D = \{2, 4, 7, 14\} \).)

\begin{table*}[!ht]%
\centering %
\caption{Restored Score Sequences\label{tab:degree_set}}%
\begin{tabular*}{\textwidth}{@{\extracolsep\fill}ll@{\extracolsep\fill}}
\toprule
\textbf{Scores} & \textbf{Exponents} \\
\midrule
$\{2,4,7,14\}$ & $\{2,1,10,3\}$ \\ \hline

$\{2,4,7,14\}$ & $\{1,2,10,4\}$ \\ \hline

$\{2,4,7,14\}$ & $\{3,1,8,5\}$ \\ \hline

$\{2,4,7,14\}$ & $\{1,1,11,5\}$ \\ \hline

$\{2,4,7,14\}$ & $\{2, 4, 5, 7\}$ \\ \hline

$\{2,4,7,14\}$ & $\{1, 2, 9, 7\}$ \\ \hline

$\{2,4,7,14\}$ & $\{3, 1, 7, 8\}$ \\ \hline

$\{2,4,7,14\}$ & $\{2, 5, 3, 9\}$ \\ \hline

$\{2,4,7,14\}$ & $\{2, 1, 8, 9\}$ \\ \hline

$\{2,4,7,14\}$ & $\{1, 5, 4, 10\}$ \\ \hline

$\{2,4,7,14\}$ & $\{3, 4, 2, 11\}$ \\ \hline

$\{2,4,7,14\}$ & $\{1, 3, 6, 11\}$ \\ \hline

$\{2,4,7,14\}$ & $\{3, 2, 4, 12\}$ \\ \hline

$\{2,4,7,14\}$ & $\{1, 3, 5, 13\}$ \\ \hline

$\{2,4,7,14\}$ & $\{3, 2, 3, 14\}$ \\ \hline

$\{2,4,7,14\}$ & $\{2, 1, 5, 15\}$ \\ \hline

$\{2,4,7,14\}$ & $\{1, 5, 1, 16\}$ \\ \hline

$\{2,4,7,14\}$ & $\{1, 2, 4, 17\}$ \\ \hline

$\{2,4,7,14\}$ & $\{3, 1, 2, 18\}$ \\ \hline

$\{2,4,7,14\}$ & $\{1, 1, 4, 19\}$ \\ \hline

$\{2,4,7,14\}$ & $\{1, 2, 1, 22\}$ \\ \hline

\bottomrule
\end{tabular*}
\end{table*}

\section{Benefits of Necessary Condition}
\label{sec:benefit2}

The necessary condition is established by enforcing that the inequality transitions into an equality when \( k = n \), yielding a computationally convenient discriminant. As the process progresses from \( k = 1 \) to \( n \), this discriminant naturally incorporates group-theoretic principles. By analyzing the minimal generating set, the solvable set, and the elements derived from prior \( x_i \), we can categorize the combinations of \( x_i \) (for \( i < k \)) into equivalence classes, thereby simplifying the computational complexity. Even without explicitly forming these equivalence classes, this method provides an \( O(n) \)-time complexity criterion for evaluating whether a given \( x_k \) has the potential to be part of a valid score sequence.  

The advantage of this condition can be observed in the examples from Table \ref{tab:fast}. For instance, in the fifth row, the reconstructed exponent set necessitates that the middle element be 2 instead of 1. Without identifying promising \( x_i \) at this stage, a straightforward brute-force search would lead to an intractable \( O(10^{26}) \) complexity, rendering the problem unsolvable.

\section{Polynomial-time algorithms}
\label{sec:pol}

In this section, we discuss polynomial-time algorithms based on Theorem \ref{thm:2.3} that reconstructs a score sequence and builds a net structure to explore all possible score sequences derived from a given score set.

\subsection{Theorems and algorithms}

To introduce the subsequent algorithm, a theorem is first proved.

\begin{theorem}\label{thm:3.1}
If a score sequence, characterized by a degree set $D$ and an exponent set $E$,  represents the score sequence of a tournament, then 
\begin{displaymath}
\forall k\in \{1, 2, \dots, n\}, p_k\leq 2\alpha_k+1, q_k\leq \alpha_k(2\alpha_k+1).
\end{displaymath}
\end{theorem}

\begin{proof}
    It can be observed that in Theorem \ref{thm:2.4}, when $k = n $, the derivation process of inequality \ref{eq:2.2f} can be applied to any $k \in \{1, 2, \dots, n\}$. Therefore, $\forall k \in \{1, 2, \dots, n\}$, 
\begin{displaymath}
1 \leq x_k \leq \frac{(2(\alpha_k-p_{k-1})+1)+\sqrt{(2\alpha_k+1)^2-8 \sum_{i=1}^{k-1} {(\alpha_k-\alpha_i)x_i}}}{2}
\end{displaymath}
\begin{displaymath}
x_k \leq \frac{(2(\alpha_k-p_{k-1})+1)+\sqrt{(2\alpha_k+1)^2-8 \sum_{i=1}^{k-1} {(\alpha_k-\alpha_i)x_i}}}{2}
\end{displaymath}
\begin{displaymath}
x_k + p_{k-1} \leq \frac{(2\alpha_k+1)+\sqrt{(2\alpha_k+1)^2-8 \sum_{i=1}^{k-1} {(\alpha_k-\alpha_i)x_i}}}{2}
\end{displaymath}
\begin{displaymath}
p_k \leq \frac{(2\alpha_k+1)+\sqrt{(2\alpha_k+1)^2-8 \sum_{i=1}^{k-1} {(\alpha_k-\alpha_i)x_i}}}{2}
\end{displaymath}
Since $\sqrt{(2\alpha_k+1)^2-8 \sum_{i=1}^{k-1} {(\alpha_k-\alpha_i)x_i}}\leq (2\alpha_k+1)$, it follows that $p_k \leq 2\alpha_k+1$. Furthermore, since $\forall i\in\{1,2,\dots,k\}, \alpha_i\leq \alpha_k$, we obtain $q_k\leq \alpha_k(2\alpha_k+1)$.
\end{proof}

From Theorem \ref{thm:2.3}, it follows that $\forall k \in \{1, 2, \dots, n\} $, the range of $x_k$ is influenced only by two variables, $p_{k-1}$ and $q_{k-1}$. Additionally, the ranges of $p_{k-1}$ and $q_{k-1}$ are limited by Theorem \ref{thm:3.1}. Therefore, by finding all possible combinations for $p_{k-1},q_{k-1}$, we can determine the exact range of $x_k$ and whether inequality \ref{eq:2.1} holds as an equality when $k=n$.

Let's begin the introduction of polynomial-time algorithm that builds a net(termed as such because, in the algorithm's flow, each potentially correct $x_i$ identified at step i is linked to some compatible candidates $x_{i-1}$ and $x_{i+1}$)  for finding all possible score sequences. With a few adjustments, this algorithm can also be used to reconstruct a score sequence in polynomial time, which will be highlighted later.

Here, we introduce the variables that will be used in the following algorithm. The variables $presize$ and $preval$ correspond to $p_{k-1}$ and $q_{k-1}$, respectively, for a given $k$. And the variable $num\_now$ represents the number of $\alpha_k$ needed, given the combination of $p_k$ and $q_k$, along with the corresponding $presize$ and $preval$. The subscripts of $size$ and $val$ at step $k$ correspond to the values of $p_k$ and $q_k$, respectively. That is, each subscript of $size$ and $val$ is associated with the values $p_k$ and $q_k$. The variable $L_k$ stores all possible combinations of $size$ and $val$ for step $k$, i.e., the combinations of $p_k$ and $q_k$.

\noindent To facilitate efficient storage and retrieval, the variable $L_k$ is organized as a nested hierarchical structure rather than a flat list of pairs. Specifically, for a fixed step $k$, we define $\text{size}_{p_k}$ as a collection that maps each valid $p_k$ to a set of associated values $\text{val}_{q_k}$. Each $\text{val}_{q_k}$, in turn, acts as a container for a set of triples $(p_{k-1}, q_{k-1}, x_k)$, representing the specific historical paths—consisting of the previous state and the current increment—that result in the combination $(p_k, q_k)$. This multi-layered nesting of sets explains why $L_k$ is represented as a complex series of nested braces in the algorithm:
\[
L_k = \left\{ p_k \rightarrow \left\{ q_k \rightarrow \left\{ (p_{k-1}, q_{k-1}, x_k)_1, (p_{k-1}, q_{k-1}, x_k)_2, \dots \right\} \right\} \right\}
\]

\noindent Furthermore, this structure naturally forms a \textbf{net}. If we treat each unique pair $(p_k, q_k)$ as a node at stage $k$, the stored triples establish directed edges connecting these nodes to their valid predecessors at stage $k-1$ and their successors at stage $k+1$. Because multiple paths can converge on or diverge from a single $(p, q)$ pair, these overlapping connections across the computational steps constitute the "net" mentioned previously.

\begin{algorithm}
\caption{The first part of building net}
\label{alg:buildnet1}
\begin{algorithmic}
\State{Define $D:=\{ \alpha_1,\ldots,\alpha_n\}$}
\State{Define $E:=\{ x_1,\ldots,x_n\}$}
\State{Define $\alpha_0:=0$}
\State{Define $\beta:=\{ \{presize\},\{preval\},\{num\_now\}$\}}
\State{Define $0\leq k\leq n$}
\State{Define $val_k:=\emptyset$}
\State{Define $size_k:=\{ \{val_0\},\ldots,\{val_{(2\alpha_k+1)\alpha_k}\}$\}}
\State{Define $L_k:=\{ \{size_0\},\ldots,\{size_{2\alpha_k+1)}\}$\}}
\State{Define $L_0.size_0.val_0:=\{ \{\{\{0,0,0\}\}\}$\}}
\State{Define $layer:=1$}
\While{$layer < n$}
\State{Define $i2:=1$}
\While{$i2\leq 2\alpha_{layer}+1$}
\State{Define $i3:=0$}
\While{$i3\leq 2\alpha_{layer-1}+1$}
\State{Define $i4:=0$}
\While{$i4\leq (2\alpha_{layer-1}+1)\alpha_{layer-1}$}
\If{$L_{layer-1}.size_{i3}.val_{i4} \neq \emptyset$}
\State{Define $h:=2(\alpha_{layer}-i3)+1$}
\If{$i2\leq \frac{h+\sqrt{h^2+8*(i4-i3(i3-1)/2)}}{2}$}
    \State{Define $\beta:=\{ i3,i4,i2\}$}
    \State{Update $L_{layer}.size_{i3+i2}.val_{i4+i2\cdot\alpha_{layer}}.insert( \{\beta\})$}
    \Comment{If we need to reconstruct score sequence, then if $L_{layer}.size_{i3+i2}.val_{i4+i2\cdot\alpha_{layer}} \neq \emptyset$, don't add $\beta$.}
\EndIf
\EndIf
\State{Update $i4 = i4 + 1$}
\EndWhile
\State{Update $i3 = i3 + 1$}
\EndWhile
\State{Update $i2=i2+1$}
\EndWhile
\State{Update $layer=layer+1$}
\EndWhile
\end{algorithmic}
\end{algorithm}
\clearpage

\addtocounter{algorithm}{-1}
\begin{algorithm}
\caption{The second part of building net}
\label{alg:buildnet2}
\begin{algorithmic}
\For{$i_3$ \textbf{from} $0$ \textbf{to} $2\alpha_{n-1} + 1$}
    \For{$i4$ \textbf{from} $0$ \textbf{to} $(2 \alpha_{n-1}+1)\alpha_{n-1}$}
        \If{$L_{n-1}.size_{i3}.val_{i4} \neq \emptyset$}
            \State{Define $h:=2(\alpha_n-i3)+1$}
            \State{Define $tmp:=\frac{h+\sqrt{h^2+8(i4-i3(i3-1)/2)}}{2}$}
            \If{$tmp \in \mathbb{N}$ \textbf{and} $tmp\geq 1$}
                \State{Define $\beta:=\{ i3,i4,tmp\}$}
                \State{Update $L_{n}.size_{i3+tmp}.val_{i4+tmp\cdot\alpha_n}.insert( \{\beta\})$}
                \Comment{If we need to reconstruct score sequence, then if $L_{n}.size_{i3+tmp}.val_{i4+tmp\cdot\alpha_{n}} \neq \emptyset$, don't add $\beta$.}
            \EndIf
        \EndIf
    \EndFor
\EndFor

\Return $L$
\end{algorithmic}
\end{algorithm}

After the computation is complete, if there exists correct score sequence, then one correct score sequence or all correct score sequences for given score set can be found by searching through the non-empty parts of $L_i.size_j.val_k$, $i$ from $n$ to $1$. Here, we show the algorithm \text{find\_all} \ref{alg:find_all} for finding all correct score sequences through the constructed net. Additionally, $cursize$ and $curval$ correspond to $p_{layer}$ and $q_{layer}$, respectively.

\addtocounter{algorithm}{-1}
\begin{algorithm}
\caption{Part for find all correct score sequences}
\label{alg:find_all}
\begin{algorithmic}
\State{Function \text{find\_all\_part}($layer,cursize,curval$)}
\If{$layer < 1$}
\Return
\Comment{At this point, all $x_i$ in $E$ have been assigned valid values; $\{E, D\}$ forms a correct score sequence.}
\EndIf
\For{each $\beta \in L_{layer}.size_{cursize}.val_{curval}$}
    \State{Update $x_{layer} := \beta.num\_now$}
    \State{$\text{find\_all\_part}(layer-1,\beta.presize,\beta.preval)$}
\EndFor

\Return

\State{Function \text{find\_all}()}
\For{$i3$ \textbf{from} $0$ \textbf{to} $2\alpha_{n}+1$}
    \For{$i4$ \textbf{from} $0$ \textbf{to} $(2\alpha_{n}+1)\alpha_{n}$}
        \If{$L_{n}.size_{i3}.val_{i4} \neq \emptyset$}
            \State{\text{find\_all\_part}($n$, $i3$, $i4$)}
        \EndIf
    \EndFor
\EndFor

\Return
\end{algorithmic}
\end{algorithm}

\begin{theorem}\label{thm:3.2}
This algorithm will complete in $O(n\alpha_n^4)$ time and requires at most $O(n\alpha_n^4)$ space. For reconstructing score sequence, it requires at most $O(n\alpha_n^3)$ space.
\end{theorem}

\begin{proof}
From the algorithm for building the net, it can be seen that the first part of the algorithm \ref{alg:buildnet1} requires only 4 nested loops. The outermost loop runs $n-1$ times. The maximum number of iterations for the second loop is $2\alpha_{n-1} + 1$, the maximum number of iterations for the third loop is $2\alpha_{n-2} + 1$, and the maximum number of iterations for the fourth loop is $(2\alpha_{n-2} + 1) \cdot \alpha_{n-2}$. The second part of the algorithm \ref{alg:buildnet2} requires only 2 nested loops, where the maximum number of iterations for the first loop is $2\alpha_{n-1} + 1$, and the maximum number of iterations for the second loop is $(2\alpha_{n-1} + 1) \cdot \alpha_{n-1}$. Therefore, the time complexity is 
\begin{subequations}
\begin{align}
    O(\text{algorithm}) 
    &= O((n-1)(2\alpha_{n-1} + 1)(2\alpha_{n-2} + 1)((2\alpha_{n-2} + 1) \alpha_{n-2})+\\&+(2\alpha_{n-1} + 1)((2\alpha_{n-1} + 1)\alpha_{n-1}))\\&=O(n\alpha_n^4)
\end{align}
\end{subequations}

Similarly, from algorithm \ref{alg:buildnet1} and algorithm \ref{alg:buildnet2}, it can be seen that the algorithm for building the net requires at most $n$ 'L' arrays, each 'L' requires at most $2\alpha_n + 1$ 'size' arrays, and each 'size' requires at most $(2\alpha_n + 1)\alpha_n$ 'val' arrays. Since the total number of operations where $\beta$ can be added to $val$ equals the total number of iterations of the algorithm, the total number of $\beta$ is $O(n\alpha_n^4)$. Therefore, the space complexity for building net is: 
\begin{subequations}
\begin{align}
    O(\text{algorithm}) 
    &= O(n(2\alpha_{n} + 1)((2\alpha_{n} + 1) \alpha_{n}))+O(n\alpha_n^4)\\&=O(n\alpha_n^4)
\end{align}
\end{subequations}

Since for the reconstruction of the score sequence, at most one $\beta$ needs to be stored for each 'val', the space complexity for reconstructing score sequence is:

\begin{subequations}
\begin{align}
    O(\text{algorithm}) 
    &= O(n(2\alpha_{n} + 1)((2\alpha_{n} + 1) \alpha_{n})\cdot 1)\\&=O(n\alpha_n^3)
\end{align}
\end{subequations}
\end{proof}

The correctness of the algorithms will now be proved. Before proving the next theorem, let us first prove an auxiliary lemma.

\begin{lemma}\label{lem:3.3}
For any $n\geq 2$ and $n > k\geq 1$, the sequence $X = \{x_i \mid x_i \in E, i \in \{1, \dots, k\}\}$ satisfies the first $k$ inequalities of Theorem \ref{thm:2.3}, if and only if information of sequence $X$ appears in the net built by algorithm \ref{alg:buildnet1}, when the outermost loop completes the $k$-th iteration.
\end{lemma}

\begin{proof}
Let's first use mathematical induction to prove this lemma.

Base case: For $k=1$, it is easy to observe from the algorithm's execution algorithm \ref{alg:buildnet1} that the algorithm will sequentially try all $x_1$ that satisfy the inequality $1 \leq x_1 \leq 2\alpha_1 + 1$, and will insert those $x_1$ that satisfy the first inequality of Theorem \ref{thm:2.3} into the 'L' array. According to Theorem \ref{thm:3.1}, the first $x_1$ that satisfies Theorem \ref{thm:2.3} must also satisfy the inequality $x_1 \leq 2\alpha_1 + 1$, so the base case is proven.

Inductive step: Assume that Lemma \ref{lem:3.3} holds for $k-1$. We will now verify that Lemma\ref{lem:3.3} holds for $k$. Assume that after algorithm \ref{alg:buildnet1} completes the step ($layer=k$) in outermost loop, there exists a sequence $X$ that satisfies the first $k$ inequalities of Theorem \ref{thm:2.3} but was not found by algorithm \ref{alg:buildnet1}. Removing $x_k$ from this sequence, by the assumption for $k-1$, the information of the newly obtained sequence $X$ is included in the net constructed after step $k-1$ of the outermost loop. The removed $x_k$ satisfies the inequality $x_k \leq 2\alpha_k + 1$ and inequality \ref{eq:2.1}. Therefore, the information of $x_k$, along with $p_{k-1}=\sum_{i=1}^{k-1} {x_i}$ of sequence $X$ and $q_{k-1}=\sum_{i=1}^{k-1} {\alpha_i x_i}$ from sequence $X$, will be stored in $L_{k}.size_{p_{k-1}+x_k}.val_{q_{k-1}+x_k\alpha_{k}}$. This leads to a contradiction. Additionally, according to the workflow of algorithm \ref{alg:buildnet1}, only the sequence $X$ that satisfies the first $k$ inequalities of Theorem \ref{thm:2.3} will be added to the net built by algorithm \ref{alg:buildnet1} after the first $k$ iterations of the outermost loop. Therefore, Lemma \ref{lem:3.3} holds for $k$.
\end{proof}

\begin{theorem}\label{thm:3.4}
A score sequence $E=\{x_i\mid x_i\in\mathbb{N},i=\in \{1,\dots,n\}\}$ is a correct score sequence if and only if the information of the score sequence $E$ appears in the net built by algorithm \ref{alg:buildnet1} and algorithm \ref{alg:buildnet2}.
\end{theorem}

\begin{proof}

First, we prove the necessity. For any correct score sequence $E$, according to Theorem \ref{thm:2.3}, we have
\begin{equation}
x_n =\frac{(2(\alpha_n-p_{n-1})+1)+\sqrt{(2(\alpha_n-p_{n-1})+1)^2+8(q_{n-1}-p_{n-1}(p_{n-1}-1)/2)}}{2},  
\end{equation}
$x_n\geq 1$, and $x_n$ is an integer. After removing $x_n$, we obtain a sequence $X$ that satisfies Lemma \ref{lem:3.3}. Therefore, the information of sequence $X$ must appear in the net constructed by algorithm \ref{alg:buildnet1}. According to the second part of the algorithm \ref{alg:buildnet2}, $p_{n-1}=\sum_{i=1}^{n-1} {x_i}$ and $q_{n-1}=\sum_{i=1}^{n-1} {\alpha_i x_i}$ of sequence $X$ must participate in the computation, since $L_{n-1}.size_{p_{n-1}}.val_{q_{n-1}}\neq \emptyset$. Therefore, the information of $x_n$ must be stored in the final net.

Now, we prove the sufficiency. If sequence $E$ appears in the constructed net, then according to the workflow of algorithm \ref{alg:buildnet1} and algorithm \ref{alg:buildnet2}, sequence $E$ must satisfy Theorem \ref{thm:2.3}, so sequence $E$ is correct score sequence.
\end{proof}

Now we prove the correctness of algorithm \ref{alg:buildnet1} and algorithm \ref{alg:buildnet2} for reconstructing the score sequence.

\begin{lemma}\label{lem:3.5}
For any $n\geq 2$ and $n > k\geq 1$, a combination of $size_{p_k}$ and $val_{q_k}$ in $L_k$ for which there exists a score sequence $X = \{x_i \mid x_i \in E, i \in \{1, \dots, k\}\}$ such that $\sum_{i=1}^{k} {x_i}=p_k$ and $\sum_{i=1}^{k} {\alpha_i x_i}=q_k$ satisfies the first $k$ inequalities of Theorem \ref{thm:2.3}, if and only if the set $L_k.size_{p_k}.val_{q_k}$ of this combination of $size_{p_k}$ and $val_{q_k}$ in $L_k$ is a non-empty set in the net built by algorithm \ref{alg:buildnet1} for reconstructing score sequence, when the outermost loop completes the $k$-th iteration.
\end{lemma}
\begin{proof}
We first use mathematical induction to prove this lemma.

Base case: For $k=1$, it follows from the execution of algorithm \ref{alg:buildnet1} for reconstructing the score sequence that the algorithm will sequentially try all $x_1$ that satisfy the inequality $1 \leq x_1 \leq 2\alpha_1 + 1$, and will insert those $x_1$ that satisfy the first inequality of Theorem \ref{thm:2.3} into the $L$ array if the place where it needs to be inserted is an empty set. According to Theorem \ref{thm:3.1}, the first $x_1$ that satisfies Theorem \ref{thm:2.3} must also satisfy the inequality $x_1 \leq 2\alpha_1 + 1$, and all combinations of $p_1$ and $q_1$ are generated by $x_1$, so the base case is proven.

Inductive step: Assume that Lemma \ref{lem:3.5} holds for $k-1$. We will now verify that Lemma \ref{lem:3.5} holds for $k$. Assume that after step $k$ (i.e., $layer=k$) of the outermost loop in algorithm \ref{alg:buildnet1} for reconstructing the score sequence is completed, there exists a combination of $size_{p_k}$ and $val_{q_k}$ in $L_k$ for which there exists a score sequence $X = \{x_i \mid x_i \in E, i = \in \{1,\dots,k\}\}$ such that $\sum_{i=1}^{k} {x_i}=p_k$ and $\sum_{i=1}^{k} {\alpha_i x_i}=q_k$ satisfies the first $k$ inequalities of Theorem \ref{thm:2.3}, but $L_k.size_{p_k}.val_{q_k}$ is an empty set. Assume that $p_{k-1}=\sum_{i=1}^{k-1} {x_i}$ and $q_{k-1}=\sum_{i=1}^{k-1} {\alpha_i x_i}$ for sequence $X$. Removing $x_k$ from this sequence, by the assumption for $k-1$, $L_{k-1}.size_{p_{k-1}}.val_{q_{k-1}}$ is non-empty set in the net constructed after step $k-1$ of the outermost loop. The removed $x_k$ satisfies the inequality $x_k \leq 2\alpha_k + 1$ and inequality \ref{eq:2.1}, so the information of $x_k$, as well as the values $p_{k-1}$ and $q_{k-1}$ of sequence $X$ will be stored in $L_{k}.size_{p_{k-1}+x_k}.val_{q_{k-1}+x_k\alpha_{k}}$ if the place where it needs to be inserted is an empty set. This contradicts our assumption, hence proving the claim. Additionally, according to the workflow of algorithm \ref{alg:buildnet1} for reconstructing score sequence, Only $L_k.size_{p_k}.val_{q_k}$ corresponding to a valid combination of $size_{p_k}$ and $val_{q_k}$ in $L_k$ for which there exists a score sequence $X = \{x_i \mid x_i \in E, i = \in \{1,\dots,k\}\}$ such that $\sum_{i=1}^{k} {x_i}=p_k$ and $\sum_{i=1}^{k} {\alpha_i x_i}=q_k$ satisfies the first $k$ inequalities of Theorem \ref{thm:2.3} will be non-empty set in the net built by algorithm \ref{alg:buildnet1} after the first $k$ iterations of the outermost loop. Therefore, Lemma \ref{lem:3.5} holds for $k$.
\end{proof}

\begin{theorem}\label{thm:3.6}
If there exists a correct score sequence $E$, then algorithm \ref{alg:buildnet1} and algorithm \ref{alg:buildnet2} for reconstructing score sequence will find at least one correct score sequence.
\end{theorem}

\begin{proof}
Suppose we have a correct sequence $E$. By removing $x_n$ from $E$, we obtain a sequence $X = \{x_i \mid x_i \in E, i \in \{1, \dots, k\}\}$. Assume that $\sum_{i=1}^{n-1} {x_i}=p_{n-1}$ and $\sum_{i=1}^{n-1} {\alpha_i x_i}=q_{n-1}$ for sequence $X$. $L_{n-1}.size_{p_{n-1}}.val_{q_{n-1}}$ will be a non-empty set in the net constructed in the first part of algorithm \ref{alg:buildnet1} for reconstructing score sequence (according to Lemma 3.5). Since the values of $p_{n-1}$ and $q_{n-1}$ are the only factors influencing the computation of
\begin{equation}
\frac{(2(\alpha_n-p_{n-1})+1)+\sqrt{(2(\alpha_n-p_{n-1})+1)^2+8(q_{n-1}-p_{n-1}(p_{n-1}-1)/2)}}{2},
\end{equation}
the second part of the algorithm \ref{alg:buildnet2} can compute the same $x_n$ using the $p_{n-1}$ and $q_{n-1}$ from $L_{n-1}.size_{p_{n-1}}.val_{q_{n-1}}$. By replacing all stored values $x_1, \dots, x_{n-1}$ in sequence $X$ with the corresponding information from the net constructed in the first part of the algorithm \ref{alg:buildnet1}, and applying a method similar to the one used in algorithm \ref{alg:find_all}, we can reconstruct a correct sequence reconstructed by algorithm \ref{alg:buildnet1} and algorithm \ref{alg:buildnet2} for reconstructing score sequence.
\end{proof}

\subsection{Experimental results}

Here I will compare the runtime of algorithm \ref{alg:buildnet1} and algorithm \ref{alg:buildnet2} with that of control algorithm \ref{alg:contro} under different scenarios, as shown in Figure \ref{fig:pol_ex_max} and Figure \ref{fig:pol_ex_num}, to verify Theorem \ref{thm:3.2}. Then, seven restored score sequences will be presented in Table \ref{tab:pol}.

\begin{algorithm}
\caption{Control Algorithm}
\label{alg:contro}
\begin{algorithmic}
\State{Define $\beta := 0$}
\For{$i1$ \textbf{from} $1$ \textbf{to} $size\_of\_score\_set$}
    \For{$i2$ \textbf{from} $1$ \textbf{to} $2 \times maximum\_of\_scores + 1$}
        \For{$i3$ \textbf{from} $1$ \textbf{to} $2 \times maximum\_of\_scores + 1$}
            \For{$i4$ \textbf{from} $1$ \textbf{to} $(2 \times maximum\_of\_scores + 1) \times maximum\_of\_scores$}
                \For{$Const := 1$ \textbf{to} 8}
                    \State{Update $\beta := 0$}
                \EndFor
            \EndFor
        \EndFor
    \EndFor
\EndFor

\Return
\end{algorithmic}
\end{algorithm}

Note: All the runtime data presented here were obtained on a standard consumer-grade laptop, a ThinkPad T14 with an i5-10210U processor, using Visual Studio C++.

\begin{figure}[tbhp]
\centering
\includegraphics[width=0.8\textwidth]{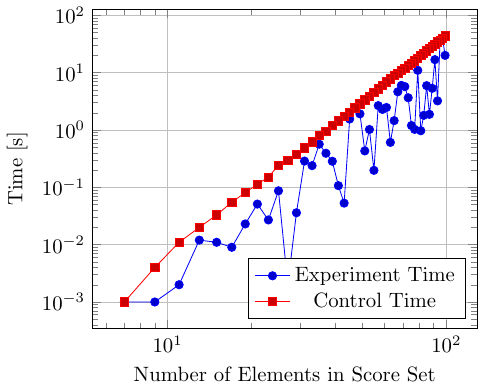}
\caption{Comparison of time measurements as the maximum element in a fixed set of scores (with a size of 7) varies, under experimental and control algorithms.}
\label{fig:pol_ex_max}
\end{figure}

\begin{figure}[tbhp]
\centering
\includegraphics[width=0.8\textwidth]{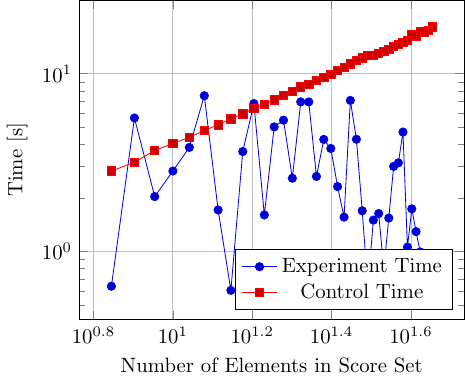}
\caption{Comparison of execution time as the number of elements in a score set varies while keeping the maximum element $\geq 28$, under experimental and control algorithms.}
\label{fig:pol_ex_num}
\end{figure}

\begin{table*}[!t]%
\centering %
\caption{Restored Score Sequences\label{tab:pol}}%
\begin{tabular*}{\textwidth}{@{\extracolsep\fill}ll@{\extracolsep\fill}}
\toprule
\textbf{Scores} & \textbf{Exponents} \\
\midrule
$\{0, 1, 2, 3, 4, 5, 7, 13, 15, 16, 17,$ & $\{1, 1, 1, 1, 1, 1, 3, 5, 2, 2, 1,$ \\
$18, 19, 20, 21, 22, 23, 24, 25, 26, 28\}$ & $1, 1, 1, 1, 1, 1, 1, 1, 1, 1\}$ \\ \hline

$\{0, 1, 2, 3, 4, 5, 6, 7, 8, 9, 12,$ & $\{1, 1, 1, 1, 1, 1, 1, 1, 1, 1, 1,$ \\
$13, 14, 15, 16, 18, 19, 20, 21, 22, 23, 24, 28\}$ & $1, 1, 1, 1, 1, 5, 3, 1, 1, 1, 1, 1\}$ \\ \hline

$\{0, 1, 2, 3, 4, 5, 6, 7, 9, 10, 11, 12, 13,$ & $\{1, 1, 1, 1, 1, 1, 1, 1, 1, 1, 1, 2, 1,$ \\
$14, 17, 18, 19, 20, 21, 22, 23, 25, 26, 27, 28\}$ & $1, 4, 1, 1, 1, 1, 1, 1, 1, 1, 1, 1\}$ \\ \hline

$\{0, 1, 2, 3, 4, 5, 6, 7, 8, 9, 10, 11, 12, 13,$ & $\{1, 1, 1, 1, 1, 1, 1, 1, 1, 1, 1, 1, 1, 1,$ \\
$14, 15, 16, 19, 20, 21, 22, 23, 25, 27, 28, 29, 30\}$ & $1, 1, 1, 1, 1, 4, 2, 1, 1, 1, 1, 1, 1\}$ \\ \hline

$\{1, 2, 4, 9, 10, 12, 13, 14, 15, 17,$ & $\{1, 1, 1, 1, 10, 2, 1, 1, 1, 1,$ \\
$18, 19, 20, 21, 23, 25, 26, 27, 28\}$ & $1, 1, 1, 1, 1, 1, 1, 1, 1\}$ \\ \hline

$\{0, 10, 11, 18, 21, 22, 26\}$ & $\{1, 6, 12, 5, 1, 1, 1\}$ \\ \hline

$\{1, 4, 5, 20, 21, 23, 28\}$ & $\{1, 5, 3, 21, 1, 1, 1\}$ \\ \hline

\bottomrule
\end{tabular*}
\end{table*}

\clearpage

\section{Fast algorithm for reconstructing score sequence}\label{sec:fast}

\subsection{Theorems and algorithms}

We now introduce the fast algorithm. This algorithm is based on Theorem \ref{thm:2.5}. To present this algorithm, we first prove the following theorem.

Assume $sol=\{s_i\mid s_i=((2\alpha_n+1-2i)i)\mod{(\alpha_n-\alpha_{n-1})}, i=\in \{0,\dots,\alpha_n-\alpha_{n-1}-1\}\}$, $grp=\{g_i\mid g_i=(\alpha_n-\alpha_i)\mod{(\alpha_n-\alpha_{n-1})}, i=\in \{1,\dots,n-1\}\}$, and $mingen=\{m_i\mid m_i=\gcd(g_i,g_{i+1},\dots,g_{n-2},g_{n-1},\alpha_n-\alpha_{n-1}),i=\in \{1,\dots,n-1\}\}$.

\begin{theorem}\label{thm:4.1}
If $\exists f\in \mathbb{N},f\leq (2\alpha_n+1)/2$, 
\begin{displaymath}
(2\alpha_n+1-2f)f=\sum_{i=1}^{n-1} {(\alpha_n-\alpha_i)x_i}, 
\end{displaymath}
then $\forall k\in \mathbb{N}, 1\leq k\leq n-2$, $sum=\sum_{i=1}^{k} {g_i x_i}$, $new\_sol=\{ns_i\mid ns_i=(s_i-sum)\mod{(\alpha_n-\alpha_{n-1})},i=\in \{0,\dots,\alpha_n-\alpha_{n-1}-1\}\}$
\begin{displaymath}
    \exists j\in \{0, 1, \dots, \alpha_n - \alpha_{n-1} - 1\}, l\in\mathbb{Z}, ns_j=l m_{k+1}.
\end{displaymath}
\end{theorem}

\begin{proof}
\begin{equation}
    ((2\alpha_n+1-2f)f)\mod(\alpha_n-\alpha_{n-1})=\gamma
\end{equation}
\begin{subequations}
\begin{align}
\gamma&=(\sum_{i=1}^{n-1} {(\alpha_n-\alpha_i)x_i})\mod(\alpha_n-\alpha_{n-1})\\
&=(\sum_{i=1}^{n-1} {((\alpha_n-\alpha_i)\mod(\alpha_n-\alpha_{n-1}))x_i})\mod(\alpha_n-\alpha_{n-1})\\
&=(\sum_{i=1}^{n-1} {g_i x_i})\mod(\alpha_n-\alpha_{n-1})\\
&=(\sum_{i=1}^{k} {g_i x_i}+\sum_{i=k+1}^{n-1} {g_i x_i})\mod(\alpha_n-\alpha_{n-1})
\end{align}
\end{subequations}
\begin{equation}
    ((2\alpha_n+1-2f)f-\sum_{i=1}^{k} {g_i x_i})\mod(\alpha_n-\alpha_{n-1})=\gamma_1
\end{equation}
\begin{equation}
    \gamma_1=(\sum_{i=k+1}^{n-1} {g_i x_i})\mod(\alpha_n-\alpha_{n-1})
\end{equation}
\begin{equation}\label{eq:4.3}
\exists \delta\in\mathbb{Z}, 
(\sum_{i=k+1}^{n-1} {g_i x_i})\mod(\alpha_n-\alpha_{n-1})=\sum_{i=k+1}^{n-1} {g_i x_i} + \delta (\alpha_n-\alpha_{n-1})
\end{equation}

Because of the properties of integer operations, 
\begin{equation}
\exists l\in\mathbb{Z}, 
\sum_{i=k+1}^{n-1} {g_i x_i} + \delta (\alpha_n-\alpha_{n-1})=l m_{k+1}
\end{equation}
\begin{equation}
    ((2\alpha_n+1-2f)f-\sum_{i=1}^{k} {g_i x_i})\mod(\alpha_n-\alpha_{n-1})=l m_{k+1}
\end{equation}
\begin{subequations}
\begin{align}
(((2\alpha_n+1-2f)f)\mod(\alpha_n-\alpha_{n-1})-\sum_{i=1}^{k} {g_i x_i})\mod(\alpha_n-\alpha_{n-1})=l m_{k+1}
\end{align}
\end{subequations}

Because $(((2\alpha_n+1-2f)f)\mod(\alpha_n-\alpha_{n-1})-\sum_{i=1}^{k} {g_i x_i})\mod(\alpha_n-\alpha_{n-1})\in new\_sol$, thus, $\exists j\in \{0, 1, \dots, \alpha_n - \alpha_{n-1} - 1\}, ns_j=l m_{k+1}$.
\end{proof}

\begin{algorithm}
\caption{The first part of reconstruction score sequence}
\label{alg:re_se1}
\begin{algorithmic}
\State{Define $D:=\{\alpha_1,\dots ,\alpha_n\}$}
\State{Define $E:=\{x_1,\dots ,x_n\}$}
\State{Define $sol=\{s_i\mid s_i=((2\alpha_n+1-2i)i)\mod{(\alpha_n-\alpha_{n-1})}, i=\in \{0,\dots,\alpha_n-\alpha_{n-1}-1\}\}$}
\State{Define $grp=\{g_i\mid g_i=(\alpha_n-\alpha_i)\mod{(\alpha_n-\alpha_{n-1})}, i=\in \{1,\dots,n-1\}\}$}
\State{Define $mingen=\{m_i\mid m_i=\gcd(g_i,g_{i+1},\dots,g_{n-2},g_{n-1},\alpha_n-\alpha_{n-1}),i=\in \{1,\dots,n-1\}\}$}

\State{Function \text{potential}($new\_sol,this\_mingen$)}
\For{each $s_i \in new\_sol$}
    \If{$(s_i==0 \textbf{ and }{this\_mingen}==0)$ or $(s_i\mod{this\_mingen}==0)$}
    \Return TRUE
    \EndIf
\EndFor

\Return FALSE

\State{Function \text{cal\_bound}($size,val,score$)}
\State{Define\\ $\gamma:=\frac{2(score-size)+1+\sqrt{(2(score-size)+1)^2+8(val-size(size-1)/2)}}{2}$}

\Return $\gamma$

\State{Function \text{end}($i,presize,preval$)}
\State{Define $this\_mod:=\alpha_n-\alpha_{n-1}$}
\State{Define $bound1:=(2\alpha_n+1)/2$}
\State{Define $bound2:=cal\_bound(presize,preval,\alpha_{n-1})$}
\While{$i\leq bound1$}
\State{Define $tmp1:=\frac{(2\alpha_n+1-2i)i-presize\cdot\alpha_n+preval}{\alpha_n-\alpha_{n-1}}$}
\If{$1\leq tmp1\leq bound2$ \textbf{and} $tmp1\in \mathbb{N}$}
\State{Define $tmp2:=cal\_bound(presize+tmp1,preval+tmp1\cdot \alpha_{n-1},\alpha_n)$}
\If{$tmp2 \in \mathbb{N}$ \textbf{and} $tmp2\geq 1$}
    \State{Update $x_n=tmp2$}
    \State{Update $x_{n-1}=tmp1$}
    \Return TRUE
\EndIf
\EndIf
\State{Update $i+=this\_mod$}
\EndWhile

\Return FALSE
\end{algorithmic}
\end{algorithm}
\clearpage

\addtocounter{algorithm}{-1}

\begin{algorithm}
\caption{The second part of reconstruction score sequence}
\label{alg:re_se2}
\begin{algorithmic}
\State{Function \text{reconstruction\_recursion}($sol\_now,layer,presize,preval$)}
\If{$layer\geq n-1$}
\For{$i:=0$ \textbf{to} $\alpha_n-\alpha_{n-1}-1$}
\If{$sol\_now.s_i==0$}
\If{$\text{end}(i,presize,preval)$}
\Return TRUE
\EndIf
\EndIf
\EndFor

\Return FALSE
\EndIf

\State{Define $bound:=cal\_bound(presize,preval,\alpha_{layer})$}

\For{$i:=1$ \textbf{to} $bound$}
\State{Define $sum:=\sum_{i=1}^{layer-1} {grp.g_i \cdot x_i}+grp.g_{layer}\cdot i$}
\State{Define $new\_sol:=\{s_i\mid s_i=(sol\_now.s_i-sum)\mod{(\alpha_n-\alpha_{n-1})}$, \\$i=\in \{0,\dots,\alpha_n-\alpha_{n-1}-1\}\}$}
\If{$\text{potential}(new\_sol,mingen_{layer+1})$}
\State{Update $x_{layer}=i$}
\If{$\text{reconstruction\_recursion}(new\_sol,layer+1,presize+i,preval+i\cdot \alpha_{layer})$}
\Return TRUE
\EndIf
\EndIf
\EndFor

\Return FALSE

\State \textbf{Function} \text{Main}()
\If{$\text{reconstruction\_recursion}(sol, 1, 0, 0) == \text{TRUE}$}
        \State \textbf{output} ``A correct score sequence has been found; the sets $D$ and $E$ form a valid sequence.''
\Else
    \State \textbf{output} ``No valid set $E$ exists for the given $D$ that can form a correct score sequence.''
\EndIf
\State \Return

\end{algorithmic}
\end{algorithm}

\begin{theorem}\label{thm:4.2}
If a correct score sequence exists, then algorithm \ref{alg:re_se1} and algorithm \ref{alg:re_se2} can reconstruct a correct score sequence.
\end{theorem}

\begin{proof}
According to Theorem \ref{thm:2.4}, Theorem \ref{thm:4.1}, and the workflow algorithm \ref{alg:re_se1} and algorithm \ref{alg:re_se2}, It can be observed that algorithm \ref{alg:re_se1} and algorithm \ref{alg:re_se2} traverse all sequences that satisfy the necessary conditions of Theorem \ref{thm:2.4} and Theorem \ref{thm:4.1}. Since a correct score sequence must be among the sequences that satisfy the necessary conditions, algorithm \ref{alg:re_se1} and algorithm \ref{alg:re_se2} are able to reconstruct the correct score sequence.
\end{proof}

\begin{theorem}\label{thm:4.3}
If algorithm \ref{alg:re_se1} and algorithm \ref{alg:re_se2} returns true in the function named 'end', then algorithm \ref{alg:re_se1} and algorithm \ref{alg:re_se2} has reconstructed a correct score sequence.
\end{theorem}

\begin{proof}
From the workflow of algorithm \ref{alg:re_se1} and algorithm \ref{alg:re_se2}, it can be seen that the sequence which makes the function named 'end' return true satisfies Theorem \ref{thm:2.3}. Therefore, when the function named 'end' returns true, algorithm \ref{alg:re_se1} and algorithm \ref{alg:re_se2} reconstructs a correct score sequence.
\end{proof}

\subsection{Experimental results}

Here, I will present 9 reconstructed score sequences along with the algorithm running times for reconstructing these sequences using algorithm \ref{alg:re_se1} and algorithm \ref{alg:re_se2} for the corresponding score sets.

Note: All runtime data presented here were obtained on a standard consumer-grade laptop, a ThinkPad T14 with an i5-10210U processor, using Visual Studio C++.

\begin{table*}[!t]%
\centering %
\caption{Restored Score Sequences\label{tab:fast}}%
\begin{tabular*}{\textwidth}{@{\extracolsep\fill}lll@{\extracolsep\fill}}
\toprule
\textbf{Scores} & \textbf{Exponents} & \textbf{Time [s]} \\
\midrule

$ \{ 351, 991, 1136, 1254, 1749, 1886, 2062, 2088, $ & $ \{ 1, 1, 1, 1, 1, 1, 1, 1, $ & 0.003 \\ $ 2387, 2632, 2804, 3827, 4034, 4115, 4153, 4436, $ & $ 1, 1, 1, 1, 1, 1, 1, 1, $ & \\ $ 4833, 5105, 5349, 5895, 6118, 6368, 6373, 7023, $ & $ 1, 1, 1, 1, 1, 1, 1, 2, $ & \\ $ 7142, 7181, 7868, 7892, 9371, 9983, \} $ & $ 1, 1, 2, 1, 3817, 15866, \} $ & \\ \hline
$ \{ 275, 552, 644, 701, 924, 1258, 1428, 1865, $ & $ \{ 1, 1, 1, 1, 1, 1, 1, 1, $ & 0.001 \\ $ 2276, 3040, 3238, 3366, 3656, 3939, 3947, 4196, $ & $ 1, 1, 1, 1, 1, 1, 1, 1, $ & \\ $ 5919, 6350, 6568, 7980, 8235, 8277, 8507, 8561, $ & $ 1, 1, 1, 1, 1, 2, 1, 1, $ & \\ $ 8993, 9026, 9131, 9365, 9398, 9548, \} $ & $ 1, 1, 1, 1, 1378, 17654, \} $ & \\ \hline
$ \{ 274, 1324, 1377, 1703, 2294, 2496, 2562, 2586, $ & $ \{ 1, 1, 1, 1, 1, 1, 1, 1, $ & 0.003 \\ $ 3140, 3150, 3251, 3467, 3703, 3773, 4134, 4274, $ & $ 1, 1, 1, 1, 1, 1, 1, 1, $ & \\ $ 4838, 5685, 6241, 6644, 6919, 7502, 7583, 7628, $ & $ 1, 1, 1, 1, 2, 1, 1, 1, $ & \\ $ 7688, 8081, 8498, 8525, 9056, 9464, \} $ & $ 1, 1, 1, 1, 1831, 16975, \} $ & \\ \hline
$ \{ 95, 146, 690, 714, 960, 1042, 1209, 1257, $ & $ \{ 1, 1, 1, 1, 1, 1, 1, 1, $ & 0.002 \\ $ 1943, 2027, 2058, 2063, 2315, 3233, 3316, 3627, $ & $ 1, 1, 1, 1, 1, 1, 1, 1, $ & \\ $ 3928, 5155, 5167, 5359, 5374, 6159, 6985, 7565, $ & $ 1, 1, 1, 1, 1, 2, 1, 1, $ & \\ $ 8435, 9030, 9070, 9300, 9350, 9620, \} $ & $ 1, 1, 1, 2, 5533, 13504, \} $ & \\ \hline
$ \{ 943, 1021, 1297, 1619, 1860, 1934, 2172, 2206, $ & $ \{ 1, 1, 1, 1, 1, 1, 1, 1, $ & 0.008 \\ $ 2282, 2532, 2654, 3576, 4518, 4593, 5693, 5715, $ & $ 1, 1, 1, 1, 1, 1, 1, 1, $ & \\ $ 7679, 8123, 9976, 11681, 12031, 12343, 12706, 12862, $ & $ 1, 1, 1, 2, 1, 1, 1, 1, $ & \\ $ 14545, 15256, 17011, 17335, 17743, 19570, \} $ & $ 1, 1, 1, 3, 35427, 13, \} $ & \\ \hline
$ \{ 568, 1586, 1610, 1621, 2357, 2538, 2542, 2971, $ & $ \{ 1, 1, 1, 1, 1, 1, 1, 1, $ & 0.005 \\ $ 3769, 3922, 5024, 5760, 8399, 8411, 9549, 9784, $ & $ 1, 1, 1, 1, 1, 1, 1, 1, $ & \\ $ 10244, 10557, 10994, 11487, 11928, 13653, 14631, 15540, $ & $ 1, 1, 2, 1, 1, 1, 1, 1, $ & \\ $ 15753, 16314, 16464, 17166, 17862, 18843, \} $ & $ 1, 1, 1, 5, 71, 37563, \} $ & \\ \hline
$ \{ 1179, 1919, 2206, 2974, 3063, 3114, 3182, 3212, $ & $ \{ 1, 1, 1, 1, 1, 1, 1, 1, $ & 0.008 \\ $ 3649, 3897, 5122, 5233, 6131, 6655, 7288, 8246, $ & $ 1, 1, 1, 1, 1, 1, 1, 1, $ & \\ $ 8880, 8927, 9220, 9331, 9970, 11095, 11713, 11728, $ & $ 1, 2, 1, 1, 1, 1, 1, 1, $ & \\ $ 12376, 12844, 15718, 16510, 17932, 19939, \} $ & $ 1, 1, 1, 1, 3031, 36493, \} $ & \\ \hline
$ \{ 701, 1021, 2614, 3755, 4387, 5158, 6863, 7828, $ & $ \{ 1, 1, 1, 1, 1, 1, 1, 1, $ & 0.008 \\ $ 8098, 8108, 8151, 8307, 8868, 8897, 10037, 10213, $ & $ 1, 1, 1, 1, 1, 1, 1, 1, $ & \\ $ 11299, 11610, 11901, 12234, 12276, 12846, 13107, 14961, $ & $ 2, 1, 1, 1, 1, 1, 1, 1, $ & \\ $ 16023, 17835, 17874, 18015, 18594, 19509, \} $ & $ 1, 1, 1, 4, 4394, 34371, \} $ & \\ \hline
$ \{ 125, 259, 769, 1293, 3514, 3754, 3769, 5258, $ & $ \{ 1, 1, 1, 1, 1, 1, 1, 1, $ & 0.006 \\ $ 5605, 6144, 6713, 6805, 6955, 7456, 7616, 9009, $ & $ 1, 1, 1, 1, 1, 1, 1, 2, $ & \\ $ 9022, 9364, 11272, 11557, 11593, 11914, 12043, 15208, $ & $ 1, 1, 1, 1, 1, 1, 1, 1, $ & \\ $ 17074, 17272, 18154, 18424, 18889, 19846, \} $ & $ 1, 1, 1, 5, 16583, 22243, \} $ & \\ \hline

\bottomrule
\end{tabular*}
\end{table*}

\clearpage

\section{Conclusion}\label{sec:con}

This paper establishes a necessary and sufficient condition and a necessary condition for reconstructing score sequences from score sets in tournament graphs, building on Landau’s theorem. These conditions define a structured set that enables the use of group-theoretic techniques, with the solution space forming a subset of this structured framework. Beyond providing a rigorous foundation for reconstruction, these conditions offer a broader perspective on related combinatorial problems, including the reconstruction of valid score sets under frequency constraints on distinct scores. Leveraging these theoretical insights, we develop three algorithms: a polynomial-time dynamic programming algorithm for reconstruction, a scalable algorithm for larger cases, and a polynomial-time network-building method for systematically identifying all possible score sequences for a given score set. These contributions bridge existing gaps in tournament graph theory and provide a structured approach to solving reconstruction problems efficiently. Furthermore, we focus on the overlap between the necessary and sufficient condition and the necessary condition, aiming to refine and derive more practical conditions that improve the efficiency and applicability of reconstruction methods.

\section*{Funding}
This research did not receive any specific grant from funding agencies in the public, commercial, or not-for-profit sectors.

\section*{Acknowledgments}

I would like to thank Professor Boris Melnikov for bringing to my attention the problem of reconstructing the score sequence from a tournament's score set.

\section*{Declaration of generative AI and AI-assisted technologies in the writing process}

During the preparation of this work the author used ChatGPT by OpenAI in order to improve language and readability, with caution. After using this tool, the author reviewed and edited the content as needed and takes full responsibility for the content of the publication.

\bibliographystyle{elsarticle-num} 
\bibliography{references}




\end{document}